\documentclass{llncs}
\usepackage{amsfonts,amssymb}
\usepackage{amsmath}
\usepackage{tikz}
\title{Polynomial Kernels for Generalized Domination Problems}
\author{Pradeesha Ashok\inst{1} \and Rajath Rao\inst{1} \and Avi Tomar\inst{1}} 

\institute{International Institute of Information Technology Bangalore, India 
\email{\{ pradeesha, rajath.rao, avi.tomar\}@iiitb.ac.in}}

\begin{document}

\newcommand{\defproblem}[3]{
  \vspace{1mm}
\noindent\fbox{
  \begin{minipage}{0.96\textwidth}
  \begin{tabular*}{\textwidth}{@{\extracolsep{\fill}}lr} #1 \\ \end{tabular*}
  {\bf{Input:}} #2  \\
  {\bf{Question:}} #3
  \end{minipage}
  }
  \vspace{1mm}
}

\newcommand{\ptds}{\textsc{Perfect total dominating set }}
\newcommand{\pds}{\textsc{Perfect dominating set} }
\newcommand{\eds}{\textsc{Efficient dominating set} }
\newcommand{\clq}{\textsc{Clique} }
\newcommand{\dpoly}{\textsc{D-Polynomial Root CSP} }
\newcommand{\sirho}{\textsc{$[\sigma , \rho ]$ Dominating Set} }
\newcommand{\dmod}{{degree-$d$-modulator}}
\newcommand{\bsirho}{\textsc{Bounded $[\sigma , \rho ]$ Dominating Set}}
\newcommand{\sbsirho}{\textsc{Semibounded $[\sigma , \rho ]$ Dominating Set}}
\newcommand{\srds}{{$(\sigma,\rho)$- dominating set}}

\maketitle
\begin{abstract}
In this paper, we study the parameterized complexity of a generalized domination problem called the $\sirho$ problem. This problem generalizes a large number of problems including the Minimum Dominating Set problem and its many variants. The parameterized complexity of the $\sirho$ problem parameterized by treewidth is well studied. Here the properties of the sets $\sigma$ and $\rho$ that make the problem tractable are identified~\cite{van2009dynamic}. We consider a larger parameter and investigate the existence of polynomial sized kernels. When $\sigma$ and $\rho$ are finite, we identify the exact condition when the \sirho problem parameterized by vertex cover admits polynomial kernels. Our lower and upper bound results can also be extended to more general conditions and provably smaller parameters as well.
\keywords{Dominating Set \and Parametrized Complexity \and Kernelization \and Sigma Rho Domination}
\end{abstract}

\section{Introduction}
Domination problems is an important topic studied in the area of Graph Theory. A well studied domination problem is the $\textsc{Minimum Dominating Set}$ problem which is stated as follows: Given a graph $G(V,E)$, find $V' \subseteq  V$ such that every vertex $v \in V\setminus V'$ has at least one neighbor in $V'$ and $|V'|$ is minimized. This problem is known to be NP-hard~\cite{gareyJohnson} and has been extensively studied in the area of approximation algorithms and exact exponential algorithms. Several variants of this problem like the $\textsc{Independent Dominating Set}$ problem, the $\textsc{Connected Dominating Set}$ problem, the $\textsc{Efficient Dominating Set}$ problem, the $\textsc{k-Dominating Set}$ problem and the $\textsc{Total Dominating Set}$ problem are also well-studied. 

We study a generalized domination problem called the $\sirho$ problem that will generalize many of the above problems. Let $\sigma$ and $\rho$ be subsets of the set of natural numbers.\\
\newline\vspace{10pt}
\defproblem{\sirho}{Graph $G(V,E)$}{Does there exist $D \subseteq V$ such that for all $v \in D$, there exists $i \in \sigma$ such that $|N(v) \cap D|=i$ and for all $v \in V\setminus D$, there exists $j \in \rho$ such that $|N(v) \cap D|=j$?}


The \sirho problem, also referred to as the Locally Checkable Vertex Subset(LCVS) problem, was introduced by Telle et al.~\cite{telle1997algorithms} in the 1990s. The \sirho problem and many of its special cases are known to be NP-hard~\cite{bange1988efficient,bange1961,efficientDom,fellows1991perfect} and therefore are unlikely to admit polynomial time exact algorithms. Therefore, it is natural to consider the complexity of parameterized versions of these problems. The $\textsc{Minimum Dominating Set}$ problem parameterized by the solution size is a W[2]-hard problem and not believed to admit FPT algorithms~\cite{downey1995parameterized}. (Formal definitions from parameterized complexity are given later). The next possibility is to consider structural parameterizations. One of the popular structural parameterizations for a graph problem is using the treewidth of the input graph. van Rooji et al.~\cite{van2009dynamic} and Telle and Proskuroswski~\cite{telle1997algorithms} studied the parameterized complexity of the \sirho problem parameterized by treewidth and proved that FPT algorithms exist when both $\sigma$ and $\rho$ are either finite or co-finite. This result implies that many domination problems are FPT when parameterized by treewidth. In addition to this, various other parameterizations of the \sirho problem are considered~\cite{jaffke2018generalized}. 

\begin{table}[t]
\begin{center}
\begin{tabular}{ |p{2cm} |p{2cm}| p{3.5cm}| c |c| p{2cm}| }
\hline
$\sigma$ & $\rho$ & Standard Name & Vertex Cover & deg. modulator & Neighborhood diversity \\ 
 \hline
$\{0\}$ &$\{1\}$& Efficient Dominating set & Poly. kernel&No Poly Kernel & Poly. kernel\\
 \hline
$\mathbb{N}$ &$\{1\}$& Perfect Dominating set & No Poly Kernel& No Poly Kernel& Poly. kernel\\
\hline
$\{1\}$ &$\{1\}$& Total Perfect Dominating set & No Poly Kernel& No Poly Kernel& Linear kernel\\ 
 \hline
$\{0\}$ &$\mathbb{N}^*$& Independent Dominating set &  Open & Open & Poly. kernel\\
\hline
$\{0,1\}$ &$\{1\}$& Weakly Perfect Dominating set & No Poly Kernel& No Poly Kernel&Poly. kernel\\
\hline
$\mathbb{N}$ &$\{i,i+1,\dots,j\}$& [i,j]-Dominating set & No Poly Kernel &  No Poly Kernel& Poly. kernel\\ 
\hline
$\{i,i+1,\dots,j\}$ &$\{i,i+1,\dots,j\}$ & Total [i,j]-Dominating set &No Poly Kernel&No Poly Kernel& Linear kernel\\
 \hline
$\{1\}$ &$\mathbb{N}^*$& Dominating Induced Matching &  Open & Open  & Poly. kernel \\ 
\hline

\end{tabular}
\end{center}
\caption{Results proved in this paper for problems expressible as $(\sigma,\rho)$ domination. Here $\mathbb{N}^*=\{1,2, \dots \}  $. Non-existence results of Polynomial kernels assume $NP \nsubseteq coNP/poly$.}
\label{table:1}
\end{table}

Once a problem is shown to admit FPT algorithms, the next interesting question is to see if the problem admits polynomial kernels. The \sirho problem parameterized by treewidth does not admit polynomial kernels unless $NP \subseteq coNP/poly$. This result can be proved using standard lower bound techniques in kernelization. Thus, to study the existence of polynomial kernels, we can consider structural parameters which are larger than the treewidth. One such parameter is the size of the vertex cover of the graph. We also consider the parameter \dmod, a generalization of vertex cover. 
\begin{definition}
Given a graph $G(V,E)$, a \dmod{} is a subset of vertices $D$ such that $G[V\setminus D]$ is a graph of maximum degree $d$.
\end{definition}
We can see that a \dmod{} with $d=0$ is the vertex cover itself. Also, the size of a \dmod{} gets possibly smaller as the value of $d$ increases. However, the size of the degree-$d$-modulator tends to be very high on many graphs including dense graphs. Therefore,  we further consider a parameter called the neighborhood diversity of the graph whose value is smaller in dense graphs. We observe that this parameter admits more positive results. Moreover, the neighborhood diversity of a graph can be computed in polynomial time~\cite{NeighbourhoodDiversity}.

\subsection{Our Results}

We consider the $\sirho$ problem with respect to the number of elements in the sets $\sigma$ and $\rho$. We are considering the minimization variants of these problems. Hence, we assume that $0 \notin \rho$, as otherwise the minimization problem of \sirho can be trivially solved by not including any vertex in the dominating set. 
\begin{enumerate}
\item Let $\sigma$ and $\rho$ be finite sets. Then the $\sirho$ problem parameterized by the size $k$ of a \dmod{} of the graph admits a polynomial kernel, if $\forall i,i \in \sigma$ implies $ i-j \notin \rho$ and $i \in \rho$ implies $i-j \notin \sigma$, for $0 \leq j \leq d$ . \vspace{5pt}  \\This implies that, when $\sigma$ and $\rho$ are finite sets, the $\sirho$ problem parameterized by the size $k$ of the vertex cover admits a polynomial kernel if $\sigma \cap \rho \ne \emptyset$. Moreover, the next result shows that this condition is tight.
\item Let $\sigma$ and $\rho$ be sets such that $\rho$ is finite and $\sigma$ is possibly infinite. Then the \sirho problem does not admit a polynomial kernel when parameterized by the size of vertex cover, if $| \sigma \cap \rho | > 0 $. 
\\ The above results show that the $[\{0\},\{1\}]$ \textsc{Domination} problem, which is commonly known as the $\textsc{Efficient Domination}$ problem or the \textsc{Perfect Code}, admits polynomial kernels when parameterized by the vertex cover. The next result shows that this cannot be extended to \dmod{} for larger values of $d$.
\item The $[\{0\},\{1\}]$ \textsc{Domination} problem does not admit polynomial kernels when parameterized by the size of degree-$1$-modulator.
\item  Let $\sigma$ and $\rho$ be finite sets. Then the \sirho problem parameterized by the neighborhood diversity admits a linear kernel.
\item  Let $\mathbb{N}^* = \mathbb{N}\setminus \{0\}$. Let $\sigma$ and $\rho$ be such that one of them is finite and the other is $\mathbb{N}^*$. Then the \sirho problem parameterized by the neighborhood diversity admits a polynomial kernel.
\end{enumerate}

Results for specific problems that follow from the above is given in Table~\ref{table:1}.
To the best of our knowledge, these results are not known before.
\section{Preliminaries}
In this section, we give definitions and results that will be used in rest of the paper.
\noindent
All the graphs considered in this paper are simple, undirected and loopless. For $n \in \mathbb{N}$, $[n]$ denotes the set $\{1,2,\dots,n\}$.
We use the standard notations from graph theory as can be found in \cite{diestel2006graph}. 
Let $G$ be a graph. We denote the vertex set of $G$ by $V(G)$ and edge set of $G$ by $E(G)$.
\vspace{5pt}
\noindent
\\\textbf{Parameterized Complexity}\cite{cygan2015parameterized}: A \textit{parameterized problem} is a language $L \subseteq \Sigma^{*} \times \mathbb{N}$, where $\Sigma$ is a fixed, finite alphabet. For an instance $(x,k) \in \Sigma^{*} \times \mathbb{N}$, k is called the parameter. 
A parameterized problem $L \subseteq \Sigma^{*} \times \mathbb{N}$ is called \textit{Fixed Parameter Tractable}(FPT) if there exists an algorithm $A$ and computable function $f:\mathbb{N} \rightarrow \mathbb{N}$ such that given $(x,k) \in \Sigma^{*} \times \mathbb{N}$, the algorithm $A$ correctly decides whether $(x,k) \in L$ in time $f(k).|(x,k)|^{O(1)}$. An important concept in FPT algorithms is kernelization, which is defined next.
\\ \noindent 
\textbf{Kernelization}\cite{cygan2015parameterized}:
A kernelization algorithm for a parameterized problem $X$, given an instance $(I,k)$ of $X$, works in polynomial time and returns an equivalent instance $(I',k')$ of $X$, where the size of the instance $(I',k')$ is bounded by $g(k)$ for some computable function $g:\mathbb{N} \rightarrow \mathbb{N}$. If $g$ is a polynomial function, then the problem is said to admit polynomial kernels. 
\\\noindent Next we define a technique that can be used to prove the possible non-existence of polynomial kernels for a parameterized problem.
\\\noindent\textbf{OR-Cross Composition} \cite{cygan2015parameterized}:
Let $L \subseteq \Sigma^*$ be a language and $Q \subseteq \Sigma^* \times \mathbb{N}$ 
 be a parameterized language. We say that $L$ \textit{OR-cross-composes} in $Q$ if there exists a polynomial equivalence relation $\mathcal{R}$ and an algorithm $A$, called the \textit{OR-cross-composition}, satisfying the following conditions. The algorithm $A$ takes as input a sequence of strings $x_1,x_2, \dots ,x_t \in \Sigma^*$ that are equivalent with respect to $\mathcal{R}$, runs in time polynomial in  $\displaystyle\sum_{i=1}^t |x_i|$, and outputs one instance $(y,k) \in \Sigma^* \times \mathbb{N}$ such that
 \begin{itemize}
 \item $k \leq p(\displaystyle\max_{i=1}^t |x_i| + \log t)$ for some polynomial $p(.)$ and
 \item $(y,k) \in Q$ if and only if there exists at least one index i such that $x_i \in L$ 
 \end{itemize}
 \begin{theorem}\label{crossth}
 \cite{bodlaender2010cross} If an NP-Hard language $L$ OR-cross-composes into a parameterized problem $Q$, then $Q$ does not admit a (generalized) polynomial kernelization unless $NP \subseteq coNP/poly$.
 \end{theorem}
 \vspace{5pt}
\noindent
\textbf{Neighborhood diversity}:
Two vertices $u,v$ in a graph $G(V,E)$, have the same type if and only if $N(v)\backslash \{u\}=N(u)\backslash \{v\}$.
If the graph vertices can be partitioned into at most $b$ sets, such that all the vertices in each partition have the same type then the graph is said to have \textit{neighborhood diversity} $b$. 

\noindent Now, we define the \dpoly problem over a field $F$  as follows~\cite{dpolycsp}.

\noindent
\defproblem{\dpoly}{A list $L$ of polynomial equalities over variable $V=\{ x_1, \dots , x_n\}$. An equality is of the form $f(x_1, \dots , x_n)=0$, where $f$ is a multivariate polynomial over $F$ of degree d.}{Does there exist an assignment of the variables $\tau : V \rightarrow \{0,1\}$ satisfying all equalities (over $F$) in $L$?}

\noindent We know the following result regarding the \dpoly.
\begin{theorem} ~\cite{dpolycsp}.\label{d-poly_th}
There is a polynomial-time algorithm that,
given an instance $(L, V)$ of \dpoly over an efficient field $F$, outputs an equivalent instance $(L',V)$ with at most $n^d + 1$ constraints such that $L' \subseteq L$.
\end{theorem}

\section{Polynomial Kernels} \label{Upper}


Let $G$ be a graph along with a \dmod{} $S \subseteq V(G)$ such that $|S| =k$.
\begin{theorem}\label{KernelThm}
Let $\sigma$ and $\rho$ be finite sets. Then the \sirho problem parameterized by the size of a \dmod{} admits a polynomial kernel if the following condition is true : $\forall i$, $i \in \sigma$ implies $i-j \notin \rho$ and $i \in \rho$ implies $ i-j \notin \sigma$, for $0 \leq j \leq d$. 
\end{theorem}

\begin{proof}
We will reduce the instance $(G,k)$ to an instance of the \dpoly problem, then carefully reduce the number of variables and then use Theorem \ref{d-poly_th} to get a polynomial kernel. A similar technique has been used in~\cite{3-cncf} to design polynomial kernels.

Let $\sigma, \rho $ be two finite sets. Given $(G,k)$, we create an instance $(L,V)$ of the \dpoly such that $L$ is satisfiable if and only if $G$ admits a \srds, $D$. Set $V := \{s_{v}  | v \in V(G) \}$. Here assignment $s_v = 1$ signifies that the vertex $v \in D$ and $s_v = 0$ signifies that $v \notin D$. 
\\\noindent We create $L$ over $\mathbb{Q}$  as follows. For every vertex $ v \in V(G) $, add the constraint
\begin{itemize}
    \item $ Y(1-s_v) + Z s_v = 0$, where 
    \begin{itemize} 
     \item $Y =\displaystyle\prod_{ i \in \rho} ((i - \sum_{u \in N(v)} s_u)^{2} + (s_v)) $
     \item $Z = \displaystyle\prod_{ j \in \sigma} ((j - \sum_{u \in N(v)} s_u)^{2} + (s_v - 1)) $
     \end{itemize}
%
    \end{itemize}

The constraint signifies that if $v \in D$, then the number of  neighbors of $v$ that belong to $D$ should be in $\sigma$ and  if $v \notin D$, then the number of  neighbors that belong to $D$ should be in $\rho$. Degree of the constraints is bounded by $2\max(|\rho| , |\sigma|) + 1$. 

\begin{lemma}
$(L,V)$ is a YES-instance if and only $G$ has a \srds.
\end{lemma} 

\begin{proof}
Let $\tau : V \rightarrow \{0, 1\}$ be a satisfying assignment for $(L, V )$. Vertex $v$ is added to a set $D$ if $\tau (s_{v}) = 1$ and $v$ is not added if $\tau (s_{v}) = 0$. 
Now we show that $D$ is a \srds. Let $v \in V(G)$ be an arbitrary vertex. If $\tau(s_v) = 1$, then there exists $i \in \sigma$ such that $i = \sum_{u \in N(v)} \tau(s_u)$. Hence $v$ has $i$ neighbors in $D$.  Similarly, if $\tau(s_v) =0$ then there exists $i \in \rho$ such that $v$ has $i$ neighbors in $D$. By symmetric arguments, we can show that if $G$ admits a \srds{} then $(L,V)$ admits a satisfiable assignment.
%
%
\qed
\end{proof}
Now $|V| = n$  where $n$ is the number of vertices in $G$. We  now modify $L$ so that it uses only $|S|$ variables. First, we make the following observation.

\begin{lemma} \label{su}

Let $ \tau$ be a satisfying assignment for $(L,V)$ and $v \in V(G)\setminus S$ and $ \displaystyle \sum_{u \in N(v) \cap S} \tau(s_u) = i - j$, where $j \in \{0,1,\dots d\}$. Then if $i \in \rho$, $\tau(s_v) =0$. Similarly, if $i \in \sigma$, then $\tau(s_v)=1$. 
\end{lemma}

\begin{proof}
Let $i \in \rho$ and $v$ be as given in the statement of the lemma.
Assume for contradiction that $\tau(s_v )= 1$. Since $\tau$ is a satisfying assignment, there exists $t \in \sigma$ such that $((t - \displaystyle \sum_{u \in N(v)} \tau(s_u))^{2} + (\tau(s_v) - 1)) = 0$. This means that  $|t - (i - j )| = \displaystyle\sum_{u \in N(v) \setminus S} \tau(s_u)$. Since $t \in \sigma $ and $i \in \rho$, we can see that $t - i$ is at least $d + 1$, hence $|t - (i - j )|  \geq |d + 1 + j|$.  This is a contradiction as $v$ can have at most $d$ neighbors in $V(G) \setminus S$. 
The second statement can be proved by similar arguments.
\qed
\end{proof}

\noindent Now, we introduce a new function.  For all $v \notin S$, let $f_{v} = g(x)$ where $x = \displaystyle\sum_{u \in N(v) \cap S} s_u$. The function $g(x)$ is defined as follows : 
\\\noindent Let $P = \displaystyle \bigcup_{i \in \rho, 0 \leq j \leq d} \{i-j \}$  and $Q = \displaystyle \bigcup_{i \in \sigma, 0 \leq j \leq d} \{i-j \}$. Now,
\begin{equation}
 g(x) = \left(\displaystyle \sum_{a\in P} \left(\prod_{c \in P \cup Q,c \ne a} \frac{x-c}{a-c}\right) + 2\displaystyle \sum_{ b \in Q } \left(\prod_{c \in P \cup Q,c \ne b}\frac{x-c}{b-c}\right)\right)-1 
 \end{equation}

Here $|P| = |\rho|(d+1)$ and $|Q|=|\sigma|(d+1)$. Therefore the degree of $g(x)$ is at most $ (d + 1) ((|\sigma| +  |\rho|) -  1)$.

Observe that for all $v \in V(G)\setminus S$, $f_v$ only uses variables defined for vertices that are in $S$. Let $L'$ be a set of constraints equal to $L$ with every occurrence of $s_v$, for $v \notin S$, substituted by $f_v$.

\begin{lemma}
If $\tau: V \rightarrow \{0, 1\}$ is a satisfying assignment for $(L, V )$ then $\tau_{|V'}$
is a satisfying assignment for $(L', V')$. Moreover, if $\tau'$ is a satisfying assignment for $(L', V')$, then there exists a satisfying assignment $\tau$ for $(L,V)$ such that $\tau' = \tau_{|V'}$.
\end{lemma}

\begin{proof}
Let $\tau$ be a satisfying assignment for $(L, V )$. Let $v\in V(G)\setminus S$ and $p=\displaystyle\sum_{u \in N(v) \cap S} \tau(s_u)$. It follows from Lemma~\ref{su} that if $p \in P$ then $\tau(s_v) = 0$ and else, if $p \in Q$, then $\tau(s_v) = 1$. We claim that  if $p \in P$ then $g(p) = 0$. This follows from the observation that exactly one term of the form $\frac{x-c}{a-c}$ becomes $1$ (when $x=a$) and every other term becomes $0$ (when $x=c$), while computing $g(p)$. Similarly, it can be observed that if $p \in Q$ then $g(p) = 1$. Thus $\tau_{|V'}$ is a satisfying assignment for $(L', V')$.

Now assume $\tau'$ is a satisfying assignment for $(L', V')$. We can now define a satisfying assignment $\tau$ for $(L,V)$ as follows: For $v \in V'$, set $\tau(v) = \tau'(v)$ and for $v \in V\setminus V'$, set $\tau(v) = f_v$. Now it follows from the above arguments that $\tau$ is is a satisfying assignment for $(L, V )$.

\qed
\end{proof}


We have obtained $(L',V')$ which is satisfiable if and only if $(L, V )$ is satisfiable. Observe that $|V'| = O(k)$ as it only contains variables corresponding to vertices in \dmod .\\

\noindent $L'$ has at most $k$ variables. Let $\gamma$ be degree of $f_v$ which is $((d + 1) (|\rho| + |\sigma|) - 1)$. Let $\alpha$ be $\max(|\rho| , |\sigma|) + 1$. \\

\noindent  We know that the maximum degree in $L'$ is $\gamma\alpha $ . We use Theorem \ref{d-poly_th} to obtain $L^{''} \subseteq L'$ such that $|L^{''}| = O(k^{(\gamma \alpha) + 1})$. 
\noindent To represent a single constraint, it is sufficient to  store the coefficient for each variable in $V'$. The storage space needed for a single coefficient  is $\log n$. Hence, $(L'', V')$ can be stored using $O(k^{(\gamma\alpha) + 2}\log{n})$ space. Now,  we can assume  $k \geq\log{n}$. Otherwise, if $\log{n} >k$, then the \dpoly problem is solvable in polynomial time by guessing the set of variables $s_v$ which will be set to $1$ in a satisfiable assignment. This can be done in $O(2^k)= O(n)$ time.

Therefore we conclude $(L'',V')$ can be stored using $O(k^{( \gamma\alpha)+2})$ bits. Now we have reduced the \sirho problem to an instance of \dpoly problem, whose size is bounded by a polynomial function of $k$. Since both these problems are in NP, we know that any instance of the \dpoly problem can be converted to an equivalent instance of the \sirho problem in polynomial time. This shows that \sirho parameterized by \dmod{} admits a polynomial kernel.
%
%
%
%
%
\end{proof}

\section{Lower Bounds}

In this section, we show lower bounds on the size of kernels for the $\sirho$ problem parameterized by vertex cover. These results complement the upper bounds proved in Section~\ref{Upper}.   We use the same framework as given in section 4.1 of \cite{3-cncf}.

\begin{theorem}
 The \sirho problem, where the set $\rho$ is finite, does not admit a polynomial kernel when parameterized by vertex cover, if $| \sigma \cap \rho | > 0 $,  unless $NP \subseteq coNP/poly$.
\end{theorem} \label{no_pol}
\begin{proof}
 We give an OR-cross-composition from \clq to the \sirho problem. Given a graph with $n$ vertices and an integer $k$, the \clq problem asks whether $G$ has a clique of size $k$. \clq is a well-known NP-hard problem~\cite{gareyJohnson}. 
We define a polynomial equivalence relation $R$ as follows.
Two instances of the \clq $(G_1 , k_1)$ and $(G_2 , k_2)$ are equivalent under $R$ if $|V(G_1)| = |V(G_2)|$ and $k_1 = k_2$. Let $t$ instances of \clq,  $G_1,\dots,G_t$, that are equivalent under $R$ be given. We arbitrarily label the vertices in each instance as $v_1 , \dots v_n$. Let $\rho$ be a fixed set that is finite and $\sigma$ be a fixed set that is possibly infinite such that $|\sigma \cap \rho| > 0$. Let $b = \min( \sigma \cap \rho ) $. We create an instance $G'$ of the \sirho problem that asks for the existence of a $[\sigma,\rho]$ dominating set of size at most $K$ ($K$ will be defined later).


The reduced instance contains three special sets of vertices referred to as $\emph{black}$, $\emph{green}$  and $\emph{blue}$ vertices. They are constructed using gadgets $\mathcal{C}$, $\mathcal{H}$ and $\mathcal{X}$ respectively. Intuitively, the construction of these vertices are such that a black vertex belongs to every dominating set and a green vertex does not belong to any dominating set. Similarly, a blue vertex may or may not belong to the dominating set based on how it is connected to the rest of the graph. We now describe the construction of these vertices using the corresponding gadgets.\\
\begin{figure}[!htb] 
    \centering
    \begin{minipage}{.5\textwidth}
        \centering
        \includegraphics[scale = 1]{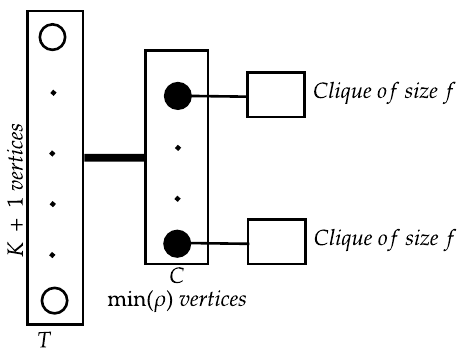}
        \caption{$\mathcal{C}$ Gadget }
        \label{c_gad}
    \end{minipage}%
    \begin{minipage}{0.5\textwidth}
        \centering
        \includegraphics[scale = 1]{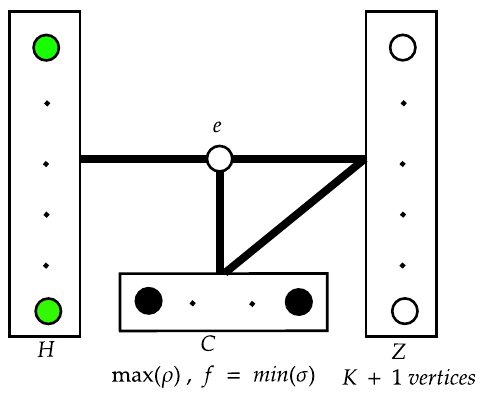}
        \caption{$\mathcal{H}$ gadget}
        
        \label{h_gad}
    \end{minipage}
\end{figure}

\noindent Construction of black vertices using the $\mathcal{C}$ gadget  : Every black vertex $v$, by virtue of its construction, has an associated bounded non-negative integer $f$, which we will refer to as the \emph{guarantee} of $v$. The significance of the value of $f$ will be explained shortly. A $\mathcal{C}$ gadget constructs black vertices in chunks of $\min (\rho)$ vertices. To construct a single chunk, we introduce a vertex set $T$ which is an independent set and $|T| = K + 1$. Similarly, let $C$ be a vertex set such that $|C| = \min(\rho)$ and $C$ is an independent set. To construct black vertices with guarantee $f$, we also add a vertex set $I$ which contains $\min(\rho)$ cliques such that each clique is of size $f$. Intuitively, the guarantee, $f$, is the number of neighbors that are added to make sure that a black vertex is dominated. Here the value of $f$ depends on where we use the black vertex. We further connect $T$ and $C$ such that $G'[T \cup C]$ is a complete bipartite graph and we connect every clique of $I$ to a unique vertex of $C$. Here the vertices in $C$ are the black vertices which are constructed. (Refer Figure~\ref{c_gad}\footnote{In all the figures, thick lines indicate all the edges across two sets are present.}).\\

We claim that all black vertices are part of every $[\sigma,\rho]$ dominating set of size at most $K$. At least one vertex $t \in T$ cannot be in the dominating set as the budget is $K$. Now $t$ needs to be dominated and as it has only $\min(\rho)$ neighbors in $C$, all the vertices of $C$ should be in the dominating set.\\

\noindent Construction of green vertices using the $\mathcal{H}$ gadget: We introduce a vertex set $Z$ where $|Z| = K + 1$ and $Z$ is an independent set, we connect every vertex of $Z$ to $\max(\rho)$ black vertices of a set $C$ which contains black vertices of guarantee $b$. We now add a vertex $e$ which is adjacent to every vertex in $Z$ and $\max(\rho)$ vertices of $C$ (same as those connected to $Z$) and to vertex set $H$. Here the vertex set $H$ is of size  at most $O(n^{2}k^{2})$  and forms an independent set. The vertices in $H$ form the constructed green vertices. (Refer Figure~\ref{h_gad}).\\

\noindent We claim that no green vertex is part of any $[\sigma,\rho]$ dominating set of size at most $K$. We observe that at least one vertex $z \in Z$ does not belong to the dominating set, as the budget is $K$. The vertex $z$ already has $\max(\rho)$ black neighbors. Hence if the vertex $e$ belongs to the dominating set then $z$ will have $\max(\rho) + 1$ neighbors in the dominating set, which is a contradiction. Hence vertex $e$ does not belong to any dominating set. Now if any green vertex of the set $H$ belongs to a dominating set then the vertex $e$ will have $\max(\rho) + 1$ neighbors in the dominating set, a contradiction. Now the claim follows.\\

\noindent Construction of blue vertices using the $\mathcal{X}$ gadget : We introduce vertices $x$ and $w$ and connect both of them to a clique of $b$ black vertices with guarantee $0$. Here $x$ is the constructed blue vertex.\\
We observe that both $x$ and $w$ are dominated irrespective of them being in the dominating set or not, by the black vertices of the clique. If $x$ is in the dominating set, the vertices of the clique are dominated, if $x$ is not in the dominating set and $\max(\rho) - 1 \notin \rho$ then to dominate the vertices in the clique, we can select $w$. We create $O(n^{2}k^{2}) $ blue vertices.    \\

As we are creating $\emph{black}$ vertices in chunks some vertices might not be connected to the rest of the construction. We connect each of such black vertices to a clique of size $b$.
Note that while describing the construction, for sake of simplicity, we will only show/refer the vertices as black, green or blue and not denote the gadgets that construct them.
Now we describe the reduction. 

We add a vertex set $P$ that contains vertices $p_{i,j}$, $i \in [k], j \in [n]$ in a grid-like fashion. These vertices will be used to select the vertices that correspond to a clique in a YES instance. Now we describe the various gadgets used in the reduction.
%
\begin{figure}
        \centering
        \includegraphics[scale = 1]{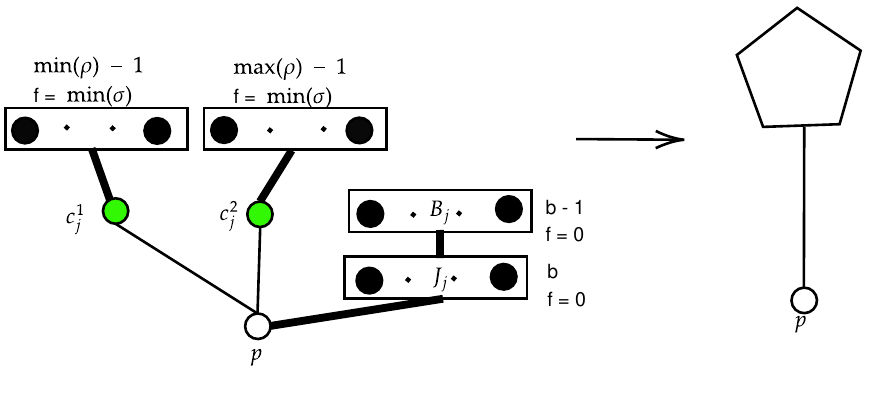}
        \caption{Row / Instance Selector Gadget }
        \label{row_gad}
\end{figure}
%
    

   \noindent \emph{Row Gadget}: For $1 \leq j \leq k$, corresponding to the $j$th row in the grid, we construct a row gadget as follows. We add two $\emph{green}$ vertices $c_{j}^{1}$ and $c_{j}^{2}$. 
    The vertex $c_{j}^{1}$ is adjacent to all the vertices, $p_{j,i}$ for $i \in [n]$, in the row and is adjacent to $\min(\rho) - 1$ black vertices each with guarantee $\min{(\sigma)}$.
    The vertex $c_{2}^{2}$ also is adjacent to all vertices in the row and has $\max(\rho) - 1$ black neighbors each with guarantee $\min{(\sigma)}$. 
    Now we add a vertex set 
    $J_j$ that contains $b$ $\emph{black}$ vertices with guarantee $0$ and a vertex set $B_j$ that contains $b - 1$ $\emph{black}$ vertices with  guarantee $0$. We add edges such that $G[B_{j} \cup J_{j}]$ is a complete bipartite graph and every vertex of $J_{j}$ is adjacent to every vertex in the row. 
\\\emph{Column Gadget}: For $1 \leq i \leq n$, we construct a column gadget for the $i$th column. We add a green vertex $u_{i}$ which is adjacent to $\max ( \rho) - 1$ black vertices with guarantee $ b$ and to all column vertices $\{p_{1,i}, \cdots, p_{k,i}\}$. We further add a blue vertex $x_i$ which is adjacent to $u_{i}$.
\\\noindent Now we can make the following claim.
\begin{claim}
A $[\sigma, \rho]$ dominating set of the reduced graph, if exists, contains exactly one vertex from each row and at most one vertex from each column in the grid.
 \end{claim}  
Let $j \in [k], i \in [n]$. To see the claim, consider the $j$th row in the grid. Since the (green) vertex $c_{j}^{1}$ in the row gadget has $\min(\rho) - 1$ black neighbors in the gadget, at least one vertex from the row should be in any $[\sigma,\rho]$ dominating set of the graph. Moreover, since the (green) vertex $c_j^2$ has $\max(\rho) - 1$ black neighbors, at most one vertex from the row can be present in any $[\sigma,\rho]$ dominating set of the graph. This implies that exactly one vertex from each row is present in any $[\sigma,\rho]$ dominating set. Similarly, for the $i$th column, since the green vertex $u_{j}$ in the column gadget has $\max(\rho) - 1$ black neighbors, at most one vertex can be selected from the $i$th column.

\noindent Now we continue our discussion of the gadgets. 
\\ \emph{Instance Selector}: We introduce a vertex set $Y$ = $\{y_1, y_2, \cdots , y_t\}$ corresponding to the $t$ instances. The vertices in $Y$ form an independent set. Intuitively, these vertices will be used to “select” the YES instance (if exists) from the $t$ instances of \clq problem. The instance selector gadget is exactly like a row gadget and it connects to $Y$ exactly like how a row gadget connects to the vertices in the row. Therefore, it follows from the previous claim that exactly one vertex from $Y$ is part of any $[\sigma,\rho]$ dominating set.
   \\\noindent\emph{Connector Gadget}: A connector gadget connects a pair of vertices from the grid to the instance selector gadget. The construction of the connector gadget depends on certain properties of the sets $\sigma$ and $\rho$, specifically whether $\max(\rho) = 1$ or $\max(\rho) > 1$ (recall that $0 \notin \rho$) and whether  $0 \in \sigma$. For every $i, i^\prime \in [k]$, $i < i' $ and $j, j^\prime \in [n]$, $j < j'$, the connector gadgets for different cases are as follows.
\begin{figure}[!htb]
    \centering
    \begin{minipage}{.5\textwidth}
        \centering
        \includegraphics[scale = 1]{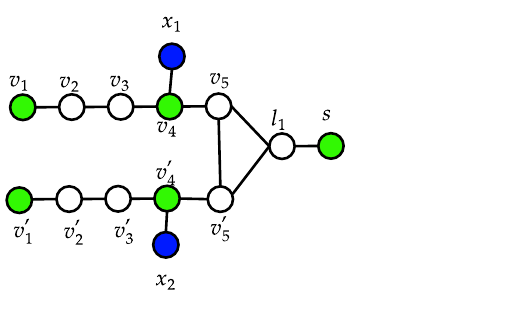}
        \caption{Connector Gadget if $\max(\rho) = 1$ and $0 \in \sigma$}
        \label{sig_0}
    \end{minipage}%
    \begin{minipage}{0.5\textwidth}
        \centering
        \includegraphics[scale = 1]{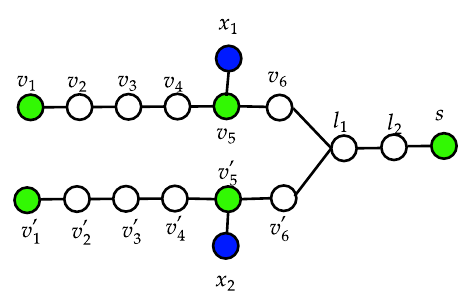}
        \caption{Connector Gadget if $\max(\rho) = 1$ and $0 \notin \sigma$}
        \label{sig_1}
    \end{minipage}
\end{figure}

   \noindent -$\max(\rho) > 1$ : We add a green vertex  $s_{i,i^\prime}^{j,j^\prime}$ and make it adjacent to $p_{i,j}$, $p_{i',j'}$,  ($\max(\rho) - 2$) black vertices with guarantee $b$ and 2 blue vertices $x_{i,i^\prime}^{j,j^\prime}$ and $l_{i,i^\prime}^{j,j^\prime}$. (Refer  Figure~\ref{complete})
    
   \noindent -$\max(\rho) = 1$ and $0 \in \sigma$ : The connector is as shown in Figure~\ref{sig_0}.
    
\noindent -$\max(\rho) = 1$ and $0 \notin \sigma$ :  The connector is as shown in Figure \ref{sig_1}.

    \noindent Finally, vertex $y_i$ has an edge to $s_{i,i^\prime}^{j,j^\prime}$ only if the edge $(i, i^\prime)$ does not exist in the graph $G_i$.
    

\begin{figure}[h]
\includegraphics[scale=0.8]{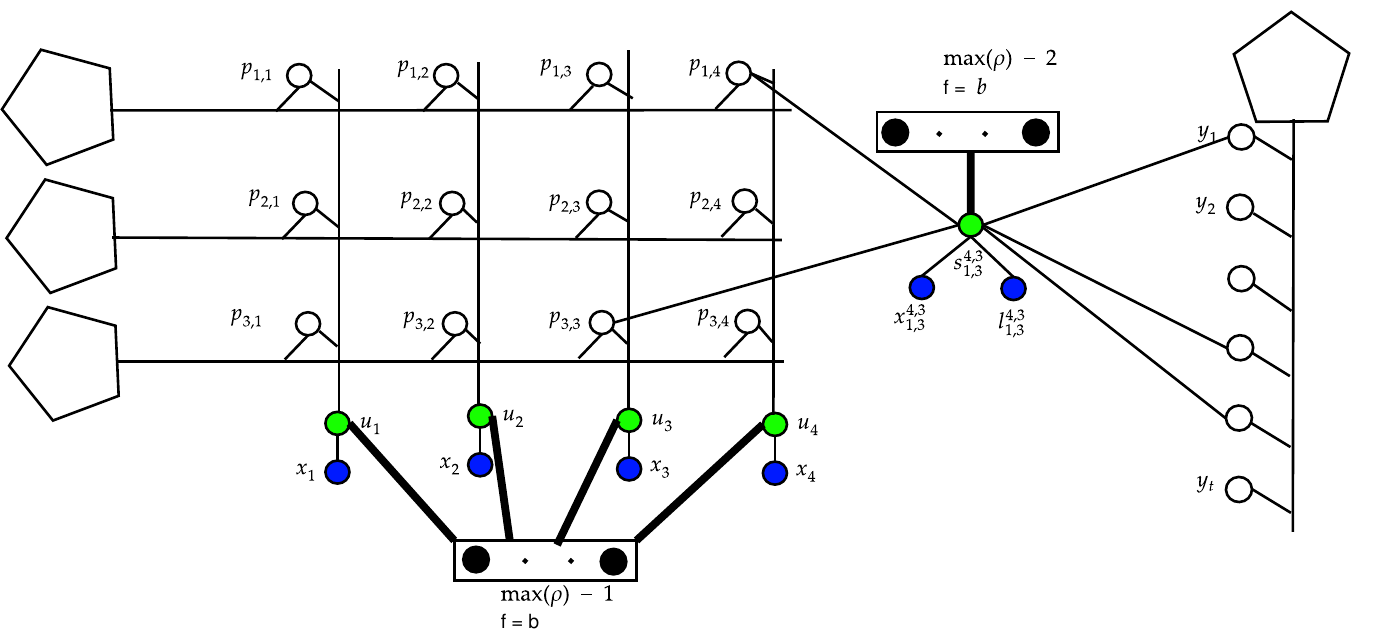}

\caption{$G$' with $n = 4$ and $k = 3$ and $\max(\rho) > 1$, In Graph 1 there is no $edge(3,4)$ hence there is an edge between $y_1$ and $s_{1,3 }^{4,3}$}
\label{complete}
\end{figure}

\noindent This concludes the description of the reduced graph, $G'$. (See  Figure~\ref{complete}).
\begin{lemma}~\label{eqvlem}
There exists a $k$-clique in one of the instances $G_i$ if and only if $G'$ has a $[\sigma, \rho]$ dominating set of size $\leq K$.
\end{lemma}

\begin{proof}


Let $D$ be a $[\sigma, \rho]$ dominating set of $G'$. We have already seen that $D$ should contain exactly one vertex from each row and at most one vertex from each column. Moreover, exactly one vertex from the set $Y$ in the instance selector, say $y_l$, is present in $D$. Let $D'$ be the indices of the set of grid vertices that are part of the dominating set $D$, i.e., $D' = \{i | p_{i,j} \in  D \}$. Now we claim that the vertices in $G_l$ corresponding to $D'$ forms a $k$-clique in $G_l$. The next lemma proves this claim. (Proof is in appendix).
%

\begin{lemma}\label{edgelemma}
If $i,i' \in D'$ then the corresponding vertices in $G_l$ have an edge between them.
\end{lemma}

It can be shown by similar arguments that if there exists a $k$-clique in one of the instances $G_l$ then  $G'$ has a $[\sigma, \rho]$ dominating set of size $\leq K$. 
\qed
\end{proof}

We set $K$ to an appropriate value, which is $n^{O(1)}$. We can see that the set $V(G') \setminus Y$ forms a vertex cover of $G'$ and $|V(G') \setminus Y| = n^{O(1)}$. Now by Theorem~\ref{crossth}, the result follows.

\qed

\end{proof}
Now we state a result for the $[\{ 0\},\{1\} ]$ Domination problem. Recall that the $[\{ 0\},\{1\} ]$ Domination problem parameterized by Degree-$0$-Modulator (i.e, the vertex cover) admits polynomial kernels, by Theorem~\ref{KernelThm}.
\begin{theorem} \label{no_pol_eff}
  The $[\{ 0\},\{1\} ]$ Domination problem, does not admit a polynomial kernel when parameterized by Degree-$1$-Modulator,  unless $NP \subseteq coNP/poly$.
 \end{theorem}
 The proof is given in appendix. 
\section{Neighborhood Diversity}
We state our results on \sirho parameterized by the neighborhood diversity. The proofs are in appendix.
\begin{theorem}
The $[\sigma,\rho]$ domination problem parameterized by neighborhood diversity, for finite $\sigma$ and $\rho$, admits a linear kernel.
\end{theorem} 
\begin{theorem}
 Let $\mathbb{N}^* = \mathbb{N}\setminus \{0\}$. The \sirho problem parameterized by neighborhood diversity admits a polynomial kernel, when $\sigma$ is finite and $\rho$ is $\mathbb{N}^*$. Also, the \sirho problem parameterized by neighborhood diversity admits a polynomial kernel, when $\rho$ is finite and $\sigma$ is $\mathbb{N}$.
 
 \end{theorem}

\nocite{*}
\bibliography{sample}
\bibliographystyle{unsrt} 
\newpage  
\section*{Appendix}
\section*{ Neighborhood Diversity}

In this section, we discuss the existence of polynomial kernels for the \sirho problem parameterized by the neighborhood diversity, under different conditions on the sets $\sigma$ and $\rho$.
Consider a connected graph $G=(V,E)$ with neighborhood diversity $b$. 
 Let $V = V_1 \uplus V_2 \uplus \dots \uplus V_b$ be the type partition of $V$.


We consider the special case of the $[\sigma, \rho]$ domination problem where the sets $\sigma$ and $\rho$ are finite and bounded by constants $s$ and $r$ respectively. We call this problem as the \textsc{Bounded $[\sigma,\rho]$ domination} problem. Now if $D$ is a $[\sigma,\rho]$ dominating set of a graph $G$, then every vertex of $G$ has at most $\max\{s,r\}$ neighbors in $D$.
\begin{theorem}
The \textsc{Bounded $[\sigma,\rho]$ domination} problem parameterized by neighborhood diversity admits a linear kernel of size $(max(s,r)+1)b$ .

\end{theorem} 

\begin{proof}
Let $G'(V',E')$ be the graph obtained from $G$ by deleting all but $\max\{s,r\}+1$ vertices from each type. (If any type contains less than $\max\{s,r\}+1$ vertices, keep all of them).  We will show that $G'$ admits a $[\sigma,\rho]$ dominating set of size $k$ if and only if $G$ admits a $[\sigma,\rho]$ dominating set of size $k$.

Let $D$ be a $[\sigma,\rho]$ dominating set of $G$ of size $k$. For $1 \leq i \leq b$, $|V_i \cap D| \leq \max\{s,r\}$. Otherwise, since $G$ is connected, there exists one vertex that is adjacent to all vertices in $V_i \cap D$ and this vertex has more than $\max\{s,r\}+1$ neighbors in $D$. Then $D' \subseteq V'$ such that $|D' \cap V_i| = |D \cap V_i|$ for $1 \leq i \leq b$ is a $[\sigma,\rho]$ dominating set of $G'$. 

Now assume that $D'$ is a $[\sigma,\rho]$ dominating set of $G'$ of size $k$. For $1 \leq i \leq b$, $|V_i \cap D'| \leq \max\{s,r\}$. Now we can show that $D'$ is a $[\sigma, \rho]$ dominating set for $G$ as well. Let $v$ be an arbitrary vertex in $V\setminus V(G')$ that belongs to type $V_i$. $V_i \cap V(G')$ has at least one vertex $u$ which is not in $D'$. Then $u$ has exactly $i$ neighbors in $D'$ for some $i \in \rho$.  Then $v$ also has exactly $i$ neighbors in $D'$, since open neighborhoods of $u$ and $v$ are the same.
\qed
\end{proof}

Now we consider the \sirho problem parameterized by neighborhood diversity where $\rho$ is bounded by a constant $r$ and $\sigma$ is $\mathbb{N}$. First we make the following observation.
\begin{lemma}\label{lm1}
Let $D$ be a minimal $[\sigma,\rho]$ dominating set. For $1 \leq i \leq b$, either $|V_i \cap D| \leq r$ or $|V_i \cap D| = |V_i|$
\end{lemma}

\begin{proof}
For contradiction, assume that $ \exists i$ such that $ r < |V_i \cap D| < |V_i|$. Since more than $r$ vertices of $V_i$ are in $D$ then all its neighbors should also be in the dominating set $D$ since every vertex not in $D$ can have at most $r$ neighbors in $D$. Therefore vertices in $V_i$ can have at most $r$ neighbors outside $V_i$, otherwise a vertex $u \in V_i \setminus D$ has more than $r$ neighbors in $D$. Now, $D' = D\setminus (V_i \cap D)$ is also a $[\sigma, \rho]$ dominating set. This contradicts that $D$ is minimal.
\qed    
\end{proof}
Note that this result implies that there exists an algorithm that runs in time $O((r+1)^bn^{O(1)})$ to solve the $\sirho$ problem. For all $1\leq i \leq b$, guess $k_i \in \{0,1,2,\dots,r,t_i\}$ such that $|V_i \cap D| =k_i$, and check if it is a valid $[\sigma,\rho]$ dominating set.  

\begin{theorem}
The \sirho problem parameterized by neighborhood diversity admits a polynomial kernel, when $\rho$ is bounded by the constant $r$ and $\sigma$ is $\mathbb{N}$.
\end{theorem}

\begin{proof}
We reduce $(G,b)$ to an instance $(G',b)$ of the \textsc{Weighted} $\sirho$ problem.

Given a graph $G(V,E)$ with a weight function $w : V\rightarrow \mathbb{N}$, the \textsc{Weighted} $\sirho$ problem is to find a $[\sigma,\rho]$ dominating set of minimum total weight.

Let $G'$ be the graph obtained from $G$ by deleting all but $r+1$ vertices from each type partition. 
In a type $V_i$, arbitrarily assign weight $1$ to $r$ vertices and assign a weight of $t_i -r$ to the remaining vertex, where $t_i= |V_i|$, for $1 \leq i \leq d$. 
If $|V_i| \leq r$ then assign weight $1$ to all the vertices of $V_i$.    

By Lemma~\ref{lm1}, we can see that $G$ admits a $[\sigma,\rho]$ dominating set of size $k$ if and only if $G'$ admits a $[\sigma,\rho]$ dominating set of total weight $k$. 
\begin{lemma}\label{lm2}
The reduced instance $G'$ can be stored in space $O(b^2)$. 
\end{lemma}
\begin{proof}
For $1 \leq i \leq b$, $|V_i| \leq r+1$ and maximum weight assigned to a vertex in $V_i$ is $t_i$. There are constant number of vertices in each type and the total number of vertices in $G'$ is at most $(r+1)b$. Now we claim that weight of every vertex can be represented using at most $b$ bits. Otherwise, $t_i > 2^b$ which implies $n >2^b$. Now the algorithm with running time $O((r+1)^bn^{O(1)})$ is $O(n^{O(1)})$ and the problem can be solved in polynomial time. 

%
 \noindent Hence, we can store the reduced instance $G'$ in $b^2$ bits. 
 \qed
\end{proof}

There exists a polynomial time many-one reduction from the \textsc{Weighted} $\sirho$ problem to the $\sirho$ problem. The result follows.

\end{proof}
%

Now we consider the \sirho problem parameterized by neighborhood diversity, when $\sigma$ is bounded by the constant $s$ and $\rho$ is $\mathbb{N}^*$. We start by proving the following result.
\begin{lemma} \label{lem2}
Let $D$ be a $[\sigma,\rho]$ dominating set. For $1 \leq i \leq b$, either $|V_i \cap D| \leq s+1$ or $|V_i \cap D| = |V_i|$.
\end{lemma}

\begin{proof}
For contradiction, assume that $ \exists i$ such that $ s+1 < |V_i \cap D| <|V_i|$. Let $u,v \in V_i$ such that $u \notin D$ and $v \in D$. Clearly, this is not possible if the type induces a clique because then $u$ has $s+1$ neighbors in $D$.  Assume, $V_i$ is an independent set. Since $ s+1 < |V_i \cap D|$, none of its neighbors are in $D$. 
This is a contradiction since $u$ does not have a neighbor in $D$.    
\end{proof}

\begin{theorem}
The \sirho problem parameterized by neighborhood diversity admits a polynomial kernel, when $\sigma$ is bounded by a constant $s$ and $\rho$ is $\mathbb{N}^*$.
\end{theorem}

\begin{proof}
We reduce $(G,b)$ to an instance of the \textsc{Weighted} \sirho problem.

Let $G'$ be the graph obtained from $G$ by deleting all but $s$ vertices from each $V_i$. Arbitrarily assign weight $1$ to $s+1$ vertices and assign a weight of $(t_i-s)+1$ to the remaining vertex, where $t_i= |V_i|$, for $1 \leq i \leq b$. If $|V_i|\leq s+1$ then we assign weight $1$ to all the vertices.    

Now the rest of the proof follows from Lemma ~\ref{lem2} using similar arguments as before.

%
%
%
\qed
\end{proof}

\section*{Proof of Lemma~\ref{edgelemma}}


We prove this for the corresponding connector in each of the three cases.


 \noindent \emph{$\max(\rho) > 1$ :}  We observe that the $\emph{green}$ vertex $s_{j,j^\prime}^{i,i^\prime}$ already has $\max(\rho) - 2$ $ \emph{black}$ neighbors.  If  $p_{i,j}$ and $p_{i',j'} \in D$, then $s_{j,j^\prime}^{i,i^\prime}$ has $\max(\rho)$ neighbors in $D$.  If $y_l \in D $ and $y_l$ has an edge to $s_{j,j^\prime}^{i,i^\prime}$ then it has $\max(\rho) + 1$ neighbors in $D$,  a contradiction.
  
\noindent  \emph{ $\max(\rho) = 1$ and $0 \in \sigma$ :} We observe that the vertices $v_3$ and $ v'_3$ belong to $D$ in order to dominate 
  $v_2$ and $ v'_2$ respectively 
  and $ l_1 \in D$ to dominate $v_5$ and $v'_5$. 
  Now $s_{j,j^\prime}^{i,i^\prime}$ will have 2 neighbors in $D$ if it is adjacent to $y_l$, a contradiction as $\max(\rho) = 1$.
  
 \noindent \emph{$\max(\rho) = 1$ and $0 \notin \sigma$}  : Since $|\sigma \cap \rho | > 0$, it follows that $1 \in \rho$ . Since $v_1$ has a neighbor in $D$, $v_4 \in D$ to dominate $v_3$ and $v_6 \notin D$. Similarly  $v'_6 \notin D$. Now to dominate $l_1$, we have to select $l_2$ in $D$.  Since $s_{j,j^\prime}^{i,i^\prime}$ already has a neighbor in $D$ it is not adjacent to $y_l$ if $y_l \in D$.

\section*{Proof of Theorem~\ref{no_pol_eff}}
 
 Here, we modify the reduction given in the proof of Theorem \ref{no_pol}.
 The construction of $\mathcal{C}$, $\mathcal{H}$, $\mathcal{X}$ gadgets are now modified as follows:
 \vspace{3pt}\\
 
 \noindent Construction of black vertices:  We introduce two vertices, $t_1$ and $t_2$, of degree $1$ and connect it to a vertex $c$. 
  Now $c$ is a $black$ vertex and is part of any dominating set. This follows from the fact that both $t_1$ and $t_2$ cannot be in the dominating set.
 
%
 \noindent Construction of green vertices : We introduce a vertex $e$ and connect it to a black vertex $c$. A vertex $h$ is connected to $e$. Now, $h$ is a green vertex and cannot be part of any dominating set, otherwise $e$ has two neighbors in the dominating set.
 
%
 \noindent Construction of blue vertices : We introduce two vertices $x$ and $l$ and connect them. Here the degree of $l$ in $G'$ is one. 
  We observe that if $x$ does not belong to the dominating set then vertex $l$ can be selected to dominate it, otherwise $l$ is dominated by $x$.\\
 \\We now describe the complete reduction.
 \begin{itemize}

 \item \emph{Instance Selector}: We introduce vertex sets $Y = \{y_1, y_2, \cdots , y_t\}$ and $Y' = \{y'_1, y'_2, \cdots , y'_t\}$ which form two independent sets. Every $y_i \in Y$ is connected to $y'_i \in Y'$. The vertices of $Y$ correspond to the $t$ instances. These vertices will be used to “select” the $G_i$ that is a YES instance.
    
    We introduce a $\emph{green}$ vertex $h$ and connect it to every vertex in the set $Y$.
    
     \item \emph{Grid :}  We add a vertex set $P$ and $B$. For $1 \leq i \leq k$ and $1 \leq j \leq n$ we add vertices $p_{i,j}$ and $b_{i,j}$ and connect them.  The vertices of the vertex set $P$ will be used to select the vertices that correspond to a clique in a YES instance. A vertex $b_{i,j} \in B$ will be colored if the corresponding vertex in $P$ is colored.
    
    \item \emph{Column Gadget}: For $1 \leq i \leq n$, we add 2 vertices  $u_{i}$ and $x_i$.  Here $u_i$ is a $\emph{green}$ vertex and is adjacent to all vertices in $\{p_{1,i}, \cdots, p_{k,i} , x_i\}$. Also, $x_i$ is a $\emph{blue}$ vertex.
    
    \item \emph{Row Gadget}: For $1 \leq j \leq k$, we add a $\emph{green}$ vertex $r_{j}$ such that 
    the vertex $r_{j}$ is adjacent to all the  vertices of the $j$th row, i.e., $p_{j,i}$ for $i \in [n]$. 
    
     \item \emph{Connector Gadget}: For every $i, i^\prime \in [k]$, $i < i' $ and $j, j^\prime \in [n]$, $j < j'$, the connector is as shown in Figure~\ref{sig_1}.
     
     \noindent Finally, vertex $y_i$ has an edge to $s_{i,i^\prime}^{j,j^\prime}$ only if the $edge(i, i^\prime)$ does not exist in the graph $G_i$.
     
      \end{itemize}
   It is easy to see that $V(G') \setminus (Y \cup Y')$ is a degree-$1$-modulator of size $O(n^2k^2)$. The next lemma completes the proof. 
      
\begin{lemma}~\label{eqvlem}
There exists a $k$-clique in one of the instances $G_i$ if and only if $G'$ has a $[\{0\}, \{1\}]$ dominating set .
\end{lemma}

\begin{proof}
Let $D$ be a $[\{0\} ,\{ 1\}]$ dominating set of $G'$. We can see that for a given $j \in [k]$, exactly one vertex from the $j$th row belongs to $D$. This is because the vertex $r_{j}$ is green and $N(r_j) = \{p_{j,1} , \dots , p_{j,k} \} $. 
 We observe that if $p_{j,i}$ is not in the dominating set, the vertex $b_{j,i}$ is in the dominating set to dominate $p_{j,i}$.  

For a given $i \in [n]$, at most one vertex from each column belongs to $D$. This is due to the presence of the green vertex $u_{i}$ such that $N(u_i) = \{p_{1,i} , \dots , p_{k,i} \}$. If none of the column vertices are selected in $D$ and the vertex $x_i$ is in $D$ to dominate $u_{i}$. By similar arguments as given for the row gadget, for a given $i \in [t]$, exactly one vertex from $Y$ belongs to $D$,  Without loss of generality, let it be $y_{l}$. Then the $k$ vertices in $G_l$ that correspond to the subset of grid vertices $T = \{i | p_{i,j} \in  D \}$ forms a $k$-clique in $G_l$. We shall prove this.

\begin{lemma}
If $i,i' \in T$ then the corresponding vertices in $G_l$ have an edge between them.
\end{lemma}

\begin{proof}

Since $v_1$ has a neighbor in $D$, $v_4 \in D$ to dominate $v_3$ and $v_6 \notin D$. Similarly  $v'_6 \notin D$. Now to dominate $l_1$, we have to select $l_2$ in $D$.  Since $s_{j,j^\prime}^{i,i^\prime}$ already has a neighbor in $D$ it is not adjacent to $y_l$, if $y_l \in D$.
\qed
\end{proof}

\noindent It can be shown by similar arguments that if there exists a $k$-clique in one of the instances $G_i$ then  $G'$ has a $[\{0\} ,\{1\}]$ Dominating set. Now Lemma~\ref{eqvlem} follows.\\
\qed
\end{proof}
\noindent We can see that the set $V(G') \setminus (Y \cup Y')$ forms a Degree 1 Modulator of $G'$ and $|V(G') \setminus (Y \cup Y')| = n^{O(1)}$ . 
Now by Theorem~\ref{crossth}, the result follows.

\end{document}